\documentclass[conference]{IEEEtran}

\usepackage{graphicx}
\usepackage{amsmath}
\usepackage{amssymb}
\usepackage{algorithm}
\usepackage{algpseudocode}

\newcommand{\Fig}[1]{Fig.~\textup{\ref{#1}}}
\graphicspath{{figures/}}

\newtheorem{lemma}{Lemma}
\newtheorem{theorem}{Theorem}
\newtheorem{corollary}{Corollary}
\newtheorem{remark}{Remark}

\newcommand{\F}{\mathbb{F}}

\begin{document}

\title{On the Multiple Threshold Decoding of LDPC codes over GF(q)}

\author{
  \IEEEauthorblockN{Alexey Frolov and Victor Zyablov}
	
  \IEEEauthorblockA{\small Inst. for Information Transmission Problems\\
    Russian Academy of Sciences\\Moscow, Russia\\
    Email: \{alexey.frolov, zyablov\}@iitp.ru
  }
}


%


\maketitle

\begin{abstract}
We consider the decoding of LDPC codes over $GF(q)$ with the low-complexity majority algorithm from \cite{FZ}. A modification of this algorithm with multiple thresholds is suggested. A lower estimate on the decoding radius realized by the new algorithm is derived. The estimate is shown to be better than the estimate for a single threshold majority decoder. At the same time the transition to multiple thresholds does not affect the order of complexity.  
\end{abstract}

\section{Introduction}
In this paper we consider the decoding of LDPC codes \cite{G, T} over $\F_q$ with the low-complexity majority algorithm from \cite{FZ}. In \cite[Theorem~1]{FZ} a lower estimate on the relative decoding radius $\rho$ realized by the low-complexity majority algorithm is derived. Let us describe the result in more detail.  Let $N$ denote the code length. In \cite{FZ} it is proved that there exist LDPC codes over $\F_q$ (with probability $p_N: \lim\nolimits_{N \to \infty} p_N \to 1$) capable of correcting any error vector of weight\footnote{Here and in what follows by weight we mean the Hamming weight, i.e. a number of non-zero elements in a vector.} $W \leq \rho N$ with the decoding complexity $O (N \log N)$. We first improve the estimate on $\rho$.

Then we consider multiple threshold decoding of LDPC codes over $\F_q$. Multiple threshold majority decoding for binary LDPC codes was first introduced in \cite{K}. In \cite{K} it was shown that transition to multiple thresholds increases the decoding radius of the majority algorithm (in the binary case the algorithm is usually called bit-flipping algorithm \cite{ZP, SS}) without affecting the order of complexity. In this paper we generalize the ideas of \cite{K} to the case of non-binary LDPC codes.

Our contribution is as follows. We first improve the estimate on the relative decoding radius $\rho$ for the single threshold case. Then we suggest the majority decoding algorithm with multiple thresholds for LDPC codes over $\F_q$. A lower estimate on the decoding radius realized by the new algorithm is derived. The estimate is shown to be at least $1.21$ times better than the estimate for a single threshold majority decoder. At the same time analogously the result from \cite{K} the transition to multiple thresholds does not affect the order of complexity. 

\section{Preliminaries}

Let us consider the construction of LDPC code $\mathcal{C}$ over $\F_q$. To construct such a code we use a bipartite graph, which is called the Tanner graph \cite{T} (see \Fig{tanner}). The graph consists of $N$  variable nodes  $v_1, v_2, \ldots, v_N$ and $M$ check nodes $c_1, c_2, \ldots, c_M$. We assume all the check nodes to have the same degree $n_0$ and all the variable nodes to have the same degree $\ell$. Such Tanner graphs are called regular ones. We associate constituent codes to each of the check nodes. All the constituent codes are the same (we denote the constituent code by $\mathcal{C}_0$). We assume $\mathcal{C}_0$ to be a linear $[n_0, R_0, d_0]$-code over $\F_q$. Let us denote the parity-check matrix of the constituent codes by $\mathbf{H}_0$. The matrix has size $m_0 \times n_0$, where $m_0 = (1-R_0)n_0$.

\begin{figure}[htbp]
\centering
\includegraphics[width=0.4\textwidth]{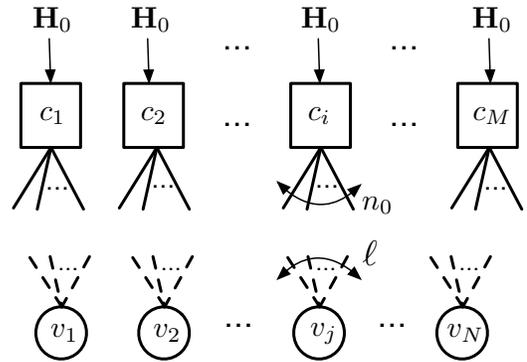}
\caption{Tanner graph}
\label{tanner}
\end{figure}

To check if $\mathbf{r} = (r_1, r_2, \ldots, r_N) \in \F_q^N$ is a codeword of $\mathcal{C}$ we associate the symbols of $\mathbf{r}$ to the variable nodes ($v_i = r_i, i = 1,\ldots, N$). The word $\mathbf{r}$ is called a codeword of $\mathcal{C}$ if all the constituent codes are satisfied (the symbols which come to the codes via the edges of the Tanner graph form codewords of the constituent codes).

It is clear the resulting code $\mathcal{C}$ is linear, so it has a parity-check matrix associated to it. We denote the matrix by $\mathbf{H}$. The code is over $\F_q$ and has the length $N$. The following inequality follows for the rate of the code $\mathcal{C}$
\[
R(\mathcal{C}) \geq 1 - \ell(1 - R_0).
\]

In what follows for the simplicity we consider only the case when the constituent code is an $[n_0, n_0-1]$ single parity-check (SPC) code over $\F_q$. The generalization to the case of a stronger constituent code is simple. It will be briefly explained in Remark~\ref{remark::gen_ldpc}.

As usually we calculate the syndrome of the sequence $\mathbf{r} = (r_1, r_2, \ldots, r_N) \in \F_q^N$ to be decoded as follows
\[
\mathbf{S} = \mathbf{H} \mathbf{r}^T.
\]

In \cite[Theorem~2]{FZ} it is proved that there exist LDPC codes over $\F_q$ (with probability $p_N: \lim\limits_{N \to \infty} p_N \to 1$) such that the following inequality holds for the syndrome weight
\begin{equation}\label{eq::syndrome}
|\mathbf{S}| > L(W) = \frac{W \ell}{2}
\end{equation}
for all error vectors of weight $W \leq W^*(R, \ell) = \omega^*(R, \ell) N$. 

To prove Theorem~2 in \cite{FZ} a Gallager-like ensemble of LDPC codes was used. The only difference to the binary case was in multiplication of the parity-check matrix columns by non-zero elements from $\F_q$. In what follows we do not need the ensemble, so we omit the definition of the ensemble here. In what follows we need just an LDPC code over $\F_q$ which satisfies the property (\ref{eq::syndrome}). We denote the code by $\mathcal{C}^*$.

We note, that at the same time the following trivial upper bound on the syndrome weight holds
\[
|\mathbf{S}| \leq U(W) = {W \ell}.
\]

\section{Single threshold majority decoding algorithm}

Let us describe a single-threshold majority decoding algorithm from \cite{FZ}. See Algorithm~\ref{alg::majority} for full description, here we give some comments and explanations. The algorithm is an iterative hard decision decoding algorithm. On each iteration the algorithm checks all the symbols from the sequence to be decoded ($\mathbf{r} = (r_1, r_2, \ldots, r_N)$). For each of the symbols the replacement criterion (see below) is checked. If the symbol satisfies the criterion, then its value is replaced with a new value, syndrome is updated and the algorithm continues with the next symbol.

\begin{remark}
It is important to note, that the algorithm works with the symbols consequently. This means, that in case of replacement all the changes are introduced to the sequence to be decoded and to the syndrome and then the algorithm goes to the next symbol.
\end{remark}

Now let us consider the replacement criterion. Assume the algorithm is considering the symbol $r_i$. The corresponding variable node $v_i$ is connected to $\ell$ constituent codes. Each of these codes sends a message to $v_i$ calculated based on values of another variable nodes connected to it (usual message passing rule). So we have $\ell$ messages coming to $v_i$. Let $A_{\max}$ denote a subset of equal non-zero messages of maximal cardinality, let $a = |A_{\max}|$ and $v$ be a value of the messages from $A_{\max}$. Let a threshold $\theta$ be an integer such that  $0 \leq \theta < \ell$. At last let $z$ be a number of zero messages. The replacement criterion is as follows. If $a - z > \theta$ we replace the symbol $r_i$ with $v$.

\begin{remark}
Note, that within the section $\theta = 0$, we introduced the parameter here just for our convenience. We will use it in the next section. 
\end{remark}

And the last thing we have not mention yet is a stopping criterion. We stop the algorithm if no changes in $\mathbf{r}$ were made during the iteration.   

\begin{algorithm}\label{alg::majority}
\caption{Single threshold majority decoding algorithm}
\begin{algorithmic}
\State {\bf Input:} received sequence $\mathbf{r}$, threshold $\theta: 0 \leq \theta < \ell$  
\State {\bf Output:} decoded sequence $\mathbf{c}$, failure flag $F$
\State {\bf Initialization:} $\mathbf{S} \gets \mathbf{H} \mathbf{r}^T$; $b \gets 1$

\While {$b = 1$}
\State $b \gets 0$
\ForAll {$1 \leq i \leq N$}
	\State calculate $\ell$ messages for $r_i$
	\State $A_{\max} \gets$ maximal subset of equal non-zero messages 
	\State $a \gets |A_{\max}|$; $v \gets$ value from $A_{\max}$
	\State $z \gets$ number of zero messages
	\If {$a - z > \theta$} 
		\State $r_i \gets v$
		\State update $\mathbf{S}$
		\State $b \gets 1$
	\EndIf
\EndFor
\EndWhile
\State $F \gets 1$
\State $\mathbf{c} \gets \mathbf{r}$
\If {$|\mathbf{S}| = 0$}
	\State $F \gets 0$
\EndIf
\end{algorithmic}
\end{algorithm}

\begin{lemma}[{\cite[Theorem 3]{FZ}}]\label{lemma::replace}
Let
\[
|\mathbf{S}| > \frac{W \ell}{2}
\]
then there exist a symbol whose replacement leads to the syndrome weight reduction (at least by $1$).
\end{lemma}
\begin{proof}
A more general proof will be given in the next section. 
\end{proof}

\begin{theorem}[{\cite[Theorem 4]{FZ}}]
Let $\mathcal{C}^*$ be an LDPC code over $\F_q$, satisfying {\textup(\ref{eq::syndrome})}. If the number of errors in the received sequence 
\[
W \leq W^*/2,
\]
the Algorithm~$1$ (with $\theta = 0$) will correct all the errors with the complexity $O (N \log N)$.  
\end{theorem}

Here we refine the result of the previous theorem
\begin{theorem}[Single threshold]
Let $\mathcal{C}^*$ be an LDPC code over $\F_q$, satisfying {\textup(\ref{eq::syndrome})}. If the number of errors in the received sequence 
\[
W \leq W^{(S)} = \frac{W^*}{2} \frac{\ell+2}{\ell+1},
\]
the Algorithm~$1$ (with $\theta = 0$) will correct all the errors with the complexity $O (N \log N)$.
\end{theorem}
\begin{proof}
To prove the theorem we need to prove that the number of errors at each step of the algorithm is less or equal to $W^*$ (see condition~(\ref{eq::syndrome}) and Lemma~\ref{lemma::replace}). 

Any error vector can mapped to a point of the following coordinate system: ``syndrome weight -- number of errors'' (see \Fig{fig::ldpc_thr_s}). At the same time it is clear, that each point in the coordinate system corresponds to multiple error vectors. First, let us add the lines $L(W)$ and $U(W)$ to \Fig{fig::ldpc_thr_s}. Recall, that the syndrome weight of any error vector with $W \leq W^*$ satisfies the inequality
\[
L(W) < |\mathbf{S}| \leq U(W).
\]

\begin{figure}[htbp]
\centering
\includegraphics[width=0.4\textwidth]{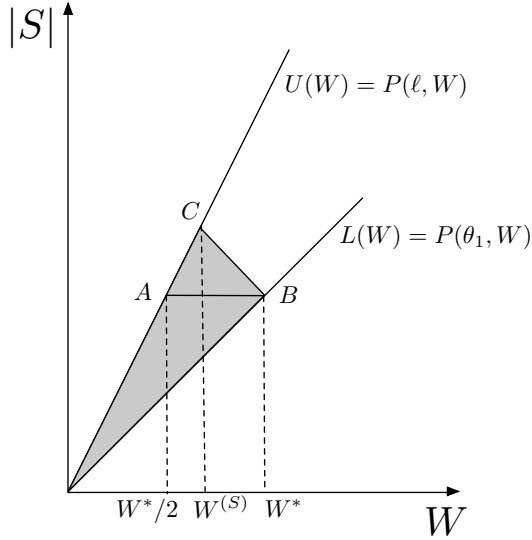}
\caption{Single threshold}
\label{fig::ldpc_thr_s}
\end{figure}

Let us consider the decoding process. It corresponds to some trajectory in the coordinate system. We start from the initial error vector. With each replacement the syndrome weight decreases (we move down at least by $1$) and the number of errors increases (we can introduce errors) or decreases by $1$ (so we move right or left by $1$). The decoding is successful if we finish at the origin. 

The area of correctable error vectors is filled by gray color in \Fig{fig::ldpc_thr_s}. Let us explain this fact. Assume we start from the point C (see \Fig{fig::ldpc_thr_s}) and only introduce errors. In this situation we move right and down by $1$ with each step (move along the line CB). We can not come to the point B as it lies on the (strict) lower bound $L(W)$ so it is clear that the number of errors can not become greater than $W^*$. In this case the decoding (and the trajectory) finishes at origin. To finish the proof we just need to calculate the coordinate of intersection of two lines: $U(W)$ and CB (starts in $W^*$ and has a slope equal to $-1$). The previous estimate ($W^*/2$, point A) is also shown in \Fig{fig::ldpc_thr_s}. 

The proof of the complexity estimate coincides with the proof from \cite{FZ}. We omit it here.
\end{proof}

\begin{corollary}
Let us introduce a notation
\[
\alpha^{(S)} = \frac{\ell+2}{2(\ell+1)} 
\]
and consider the asymptotic ($N \to \infty$) estimate of the relative decoding radius realized by Algorithm~$1$. We have
\[
\rho^{(S)} \geq \frac{W^{(S)}}{N} = \alpha^{(S)}  \omega^*.
\]
\end{corollary}

In the next section we will increase the estimate by means of transition to multiple decoding thresholds.

\section{Decoding with multiple thresholds}

Let us first introduce the sequence of integer thresholds (let $t \geq 1$)
\[
0 = \theta_1 < \theta_2 < \ldots < \theta_t < \ell. 
\]

Now we are ready to describe the multiple threshold decoding algorithm. The idea of the new algorithm is in consequent applying the Algorithm~$1$ with different replacement thresholds to the sequence to be decoded. We start from the largest threshold $\theta_t$ and end with $\theta_1 = 0$. Please see Algorithm~$2$ full description below for more details.

\begin{algorithm}\label{alg::majority_m}
\caption{Multiple threshold majority decoding algorithm}
\begin{algorithmic}
\State {\bf Input:} received sequence $\mathbf{r}$, $t$ thresholds $0 = \theta_1 < \theta_2 < \ldots < \theta_t < \ell$  
\State {\bf Output:} decoded sequence $\mathbf{c}$, failure flag $F$
\State {\bf Initialization:} $\mathbf{S} \gets \mathbf{H} \mathbf{r}^T$

\ForAll {$0 \leq i \leq t-1$}
	\State Apply Algorithm~$1$ with $\theta = \theta_{t-i}$
	\State $\mathbf{r} \gets$ output of Algorithm~$1$
\EndFor

\State $F \gets 1$
\State $\mathbf{c} \gets \mathbf{r}$
\If {$|\mathbf{S}| = 0$}
	\State $F \gets 0$
\EndIf
\end{algorithmic}
\end{algorithm}

\begin{remark}
We note, that the implementation of the Algorithm~$2$ is not optimal. It is much better to implement it in such a way. First calculate the syndrome, then sort all the symbols in a descending order of $a - z$ value (see previous section), then change the symbols consequently and update the sorted list. But nevertheless we see here that the complexity of Algorithm~$2$ is no more than $t$ times the complexity of Algorithm~$1$. So the order of complexity is $O(N \log N)$.    
\end{remark}

To estimate the decoding radius of the Algorithm~$2$ we need the following Lemma.

\begin{lemma}
Let $\theta$ be an integer, $0 \leq \theta < \ell$, let
\[
|\mathbf{S}| > P(\theta, W) = W \frac{\ell + \theta}{2}
\]
then there exist a symbol whose replacement leads to the syndrome weight reduction by at least by $\theta + 1$.
\end{lemma}
\begin{proof}
Consider a subgraph of the Tanner graph that contains only erroneous symbols (the number of errors is equal to $W$) and constituent codes connected to these symbols. Within the proof we work with this subgraph only.

Let us introduce the following notation:
\begin{itemize}
\item $A$ is the set of codes that detect an error ($|A| = |\mathbf{S}|$);
\item $A_i$, $i = 1,\ldots,n_0$, is the subset of  $A$ containing only the codes with precisely $i$ incoming edges ($a_i = |A_i|$);
\item $A_{\geq 2} = A \backslash  A_1$ is a subset of  $A$ containing only the codes with at least $2$ incoming edges ($a_{\geq 2}  = |A_{\geq 2}|$);
\item $C$ is the set of codes that contain errors but do not detect them ($c = |C|$);
\item $e^{(i)}_{A_1}$ is the number of edges outgoing from a symbol $i$ and incoming to $A_1$;
\item $e^{(i)}_C$ is the number of edges outgoing from a symbol $i$ and incoming to $C$.
\end{itemize}

\begin{figure}[htbp]
\centering
\includegraphics[width=0.4\textwidth]{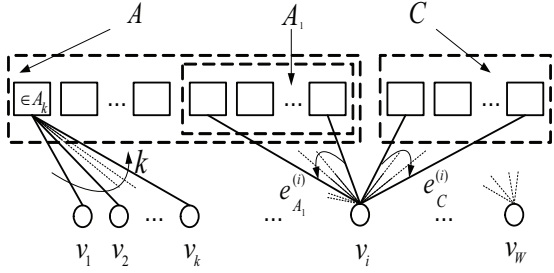}
\caption{A subgraph of Tanner graph}
\label{fig::tanner_sub}
\end{figure}

In \Fig{fig::tanner_sub} we present an example of a subgraph of the Tanner graph and illustrate the introduced notation.

First note, that if the condition 
\[
e^{(i)}_{A_1} > e^{(i)}_C + \theta
\]
holds for the $i$-th symbol, then the replacement of it will lead to the syndrome weight reduction by at least by $\theta + 1$. To prove this it is sufficient to mention that the codes with the only error will give equal messages.

Then we claim that if
\[
a_1 > \sum\limits_{i=1}^{W} e^{(i)}_C + W \theta,
\]
then there exist a symbol $i$ such that $e^{(i)}_{A_1} > e^{(i)}_C + \theta$.

And to finish the proof we need to count the edges in the subgraph. The number of edges outgoing from $W$ erroneous symbols is $W \ell$. These
edges can come to either codes that have detected an error ($A = A_1 \cup A_{\geq 2}$) or to codes that have not detected errors but contain them ($C$). Let us estimate the number of edges incoming to each of the three sets of codes:
\begin{itemize}
\item The number of edges leading to codes of the set $A_1$ is $\sum\limits_{i=1}^{W}e^{(i)}_{A_1} = a_1$;
\item The number of edges leading to codes of the set $A_{\geq 2}$ is at least $2(|\mathbf{S}| - a_1)$ (here we use the fact every code has at least two incoming edges);
\item The number of edges leading to codes of the set $C$ is $\sum\limits_{i=1}^{W}e^{(i)}_{C}$.
\end{itemize}

Thus
\[
W \ell \geq a_1 + 2(|\mathbf{S}| - a_1) + \sum\limits_{i=1}^{W}e^{(i)}_{C}.
\]

After some transformations, we have
\[
a_1 - \sum\limits_{i=1}^{W} e^{(i)}_C \geq 2 |\mathbf{S}| - W \ell.
\]

This immediately implies that if the condition of the Lemma holds then 
\[
a_1 > \sum\limits_{i=1}^{W} e^{(i)}_C + W \theta.
\]
\end{proof}

\begin{theorem}[Multiple thresholds]
Let $\mathcal{C}^*$ be an LDPC code over $\F_q$, satisfying ($\ref{eq::syndrome}$). Let $0 = \theta_1 < \theta_2 < \ldots < \theta_t < \ell$ be a sequence of thresholds. If the number of errors in the received sequence 
\[
W \leq W_{t+1},
\]
where
\[
W_{i} = W_{i-1} \frac{\ell + 3 \theta_{i-1} + 2}{\ell + 2 \theta_{i-1} + \theta_{i} + 2}, \: W_1 = W^*, \theta_{t+1} = \ell,
\]
the Algorithm~$2$ will correct all the errors with complexity $O(N \log N)$.
\end{theorem}
\begin{proof}
The area of correctable error vectors is shown in \Fig{fig::ldpc_thr_m}. For now the area is more difficult because the slope at threshold $\theta_i$ is equal to $\theta_i +1$. To prove the Theorem we need to consequently calculate coordinates of intersection of the area bound and lines $P(\theta_i, W)$.

\begin{figure}[htbp]
\centering
\includegraphics[width=0.4\textwidth]{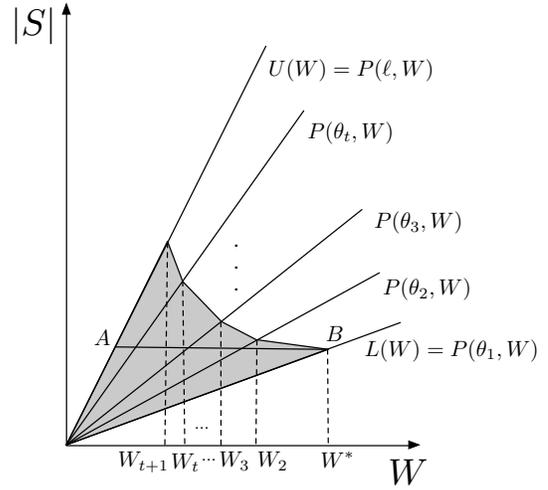}
\caption{Multiple thresholds}
\label{fig::ldpc_thr_m}
\end{figure}

\end{proof}

The most interesting case for us is the case when we have all the thresholds from $0$ to $\ell-1$. In this case
\[
W^{(M)} = \prod\limits_{i=0}^{\ell-1} \frac{\ell+3i+2}{\ell+3i+3} W^*.
\]
Let us introduce a notation
\[
\alpha^{(M)} = \prod\limits_{i=0}^{\ell-1} \frac{\ell+3i+2}{\ell+3i+3}
\]
and consider the asymptotic ($N \to \infty$) estimate of the relative decoding radius realized by Algorithm~$2$ (when we have all the thresholds). We have
\[
\rho^{(M)} \geq \frac{W^{(M)}}{N} = \alpha^{(M)} \omega^*.
\]

In \Fig{fig::alpha_m} the comparison of $\alpha^{(S)}$ and $\alpha^{(M)}$ is shown.

\begin{figure}[htbp]
\centering
\includegraphics[width=0.35\textwidth]{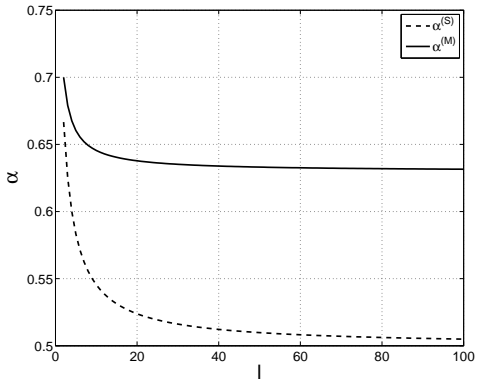}
\caption{The dependency of $\alpha^{(S)}$ and $\alpha^{(M)}$ on $\ell$}
\label{fig::alpha_m}
\end{figure}

At last let us calculate the value of $\alpha^{(M)}$ when $\ell$ is big. It is easy to check, that
\[
\lim\limits_{\ell \to \infty} \alpha^{(M)} = 2^{-2/3} = 0.6300...
\]

\begin{remark}[Generalized LDPC codes]\label{remark::gen_ldpc}
Here we briefly consider the case of generalized LDPC codes, i.e. the case when the constituent codes are not SPC codes but some more powerful codes. All our theorems work in this case if we use the so-called generalized syndrome rather then an ordinary syndrome. Generalized syndrome consists of syndromes of constituent codes. The weight of generalized syndrome is just the number of unsatisfied constituent codes. We would like to point out, that analogously to \cite{FZ} the transition to generalized LDPC codes does not lead to a gain in the decoding radius. 
\end{remark}

\section{Numerical results}

The numerical results are given in Table~\ref{t_16} for $q=16$ and Table~\ref{t_64} for $q=64$. In each Table the dependencies of $\omega^*$, $\rho^{(S)}$ and $\rho^{(M)}$ on the code rate $R$ are presented. Note, that $\ell$ (in each case) is chosen to maximize the functions. For our case the maximal values of $\omega^*$, $\rho^{(S)}$ and $\rho^{(M)}$ were achieved for the same $\ell$, the value of $\ell$ is also given in the Tables.
   

\begin{table}[!t]
\caption{Results for $q=16$}
\label{t_16}
\centering
\begin{tabular}{| l | c | c | c |}
  \hline                       
  $R$; $\ell$ &  $\omega^*$ & $\rho^{(S)}$ & $\rho^{(M)}$\\
  \hline
  0.125; 45 & 0.0103 & 0.0053 & 0.0065 \\
  0.25; 43   & 0.0095 & 0.0049 & 0.0060 \\
  0.375; 40 & 0.0085 & 0.0044 & 0.0054 \\
  0.5; 31     & 0.0072 & 0.0037 & 0.0046 \\
  0.625; 24 & 0.0053 & 0.0028 & 0.0034 \\
  0.75; 24   & 0.0033 & 0.0017 & 0.0021 \\
   0.875; 26 & 0.0015 & 0.0008 & 0.0010 \\
  \hline  
\end{tabular}
\end{table}

\begin{table}[!t]
\caption{Results for $q=64$}
\label{t_64}
\centering
\begin{tabular}{| l | c | c | c |}
  \hline                       
  $R$; $\ell$ &  $\omega^*$ & $\rho^{(S)}$ & $\rho^{(M)}$\\
  \hline
  0.125; 21 & 0.0156 & 0.0082 & 0.0099 \\
  0.25; 24   & 0.0131 & 0.0068 & 0.0083 \\
  0.375; 20 & 0.0104 & 0.0054 & 0.0066 \\
  0.5; 22     & 0.0081 & 0.0042 & 0.0052 \\
  0.625; 27 & 0.0059 & 0.0031 & 0.0038 \\
  0.75; 24   & 0.0037 & 0.0019 & 0.0024 \\
   0.875; 26 & 0.0017 & 0.0009 & 0.0011 \\
  \hline  
\end{tabular}
\end{table}

\newpage
We note, that the value of $\rho^{(M)}/\rho^{(S)} \geq 1.21$ for all the rates we considered. So transition to multiple thresholds leads to the gain in the decoding radius without affecting the order of complexity. To the best knowledge of the authors the obtained estimates are currently the best estimates of the decoding radius for low-complexity majority decoder of LDPC codes over $\F_q$.

\section{Conclusion}
We improved the estimate on the relative decoding radius $\rho$ for the single threshold majority decoder of LDPC codes over $\F_q$. The majority decoding algorithm with multiple thresholds is suggested. A lower estimate on the decoding radius realized by the new algorithm is derived. The estimate is shown to be at least $1.21$ times better than the estimate for a single threshold majority decoder. At the same time analogously the result from \cite{K} the transition to multiple thresholds does not affect the order of complexity.

All the results are obtained for the case when the constituent codes are SPC codes over $\F_q$. The case of more powerful constituent codes is considered. It is shown that analogously to \cite{FZ} the transition to generalized LDPC codes does not lead to a gain in the decoding radius. 

To the best knowledge of the authors the obtained estimates are currently the best estimates of the decoding radius for low-complexity majority decoder of LDPC codes over $\F_q$.

\section*{Acknowledgment}
This work was partially supported by Russian Science Foundation grant 14-50-00150.

\end{document}